\documentclass[11pt]{article}

\usepackage{mathpazo}
\usepackage{amsfonts}
\usepackage{amsmath}
\usepackage{latexsym}

\usepackage[margin=1.05in]{geometry}
\usepackage[pagebackref]{hyperref}
\hypersetup{pdfpagemode=UseNone}

\usepackage{amsthm}
\newtheorem{theorem}{Theorem}
\newtheorem{lemma}[theorem]{Lemma}
\newtheorem{prop}[theorem]{Proposition}

\theoremstyle{definition}

\newtheorem{remark}[theorem]{Remark}

\usepackage{xcolor}

\newcommand{\tinyspace}{\mspace{1mu}}
\newcommand{\microspace}{\mspace{0.5mu}}
\newcommand{\negsmallspace}{\mspace{-1.5mu}}
\newcommand{\op}[1]{\operatorname{#1}}
\newcommand{\tr}{\operatorname{Tr}}
\newcommand{\pt}{\operatorname{T}}

\renewcommand{\int}{\operatorname{int}}

\renewcommand{\vec}{\operatorname{vec}}

\renewcommand{\t}{{\scriptscriptstyle\mathsf{T}}}

\newcommand{\ip}[2]{\langle #1 , #2\rangle}
\newcommand{\bigip}[2]{\bigl\langle #1, #2 \bigr\rangle}

\newcommand{\norm}[1]{\lVert\tinyspace #1 \tinyspace\rVert}

\newcommand{\ket}[1]{
  \lvert\microspace #1 \microspace \rangle}

\newcommand{\bra}[1]{
  \langle\microspace #1 \microspace \rvert}

\def\I{\mathbb{1}}

\newcommand{\setft}[1]{\mathrm{#1}}
\newcommand{\Density}{\setft{D}}
\newcommand{\Pos}{\setft{Pos}}
\newcommand{\PPT}{\setft{PPT}}
\newcommand{\Unitary}{\setft{U}}
\newcommand{\Herm}{\setft{Herm}}
\newcommand{\Lin}{\setft{L}}
\newcommand{\Sep}{\setft{Sep}}
\newcommand{\BPos}{\setft{Sep}^{\ast}}

\def\complex{\mathbb{C}}
\def\real{\mathbb{R}}

\newenvironment{mylist}[1]{\begin{list}{}{
	\setlength{\leftmargin}{#1}
	\setlength{\rightmargin}{0mm}
	\setlength{\labelsep}{2mm}
	\setlength{\labelwidth}{8mm}
	\setlength{\itemsep}{0mm}}}
	{\end{list}}

\def\X{\mathcal{X}}
\def\Y{\mathcal{Y}}
\def\Z{\mathcal{Z}}
\def\W{\mathcal{W}}
\def\A{\mathcal{A}}
\def\B{\mathcal{B}}

\def\P{\mathcal{P}}
\def\S{\mathcal{S}}

\def\K{\mathcal{K}}

\def\eps{\varepsilon}

\makeatletter
\let\@fnsymbol\@arabic
\makeatother

\begin{document}

\title{\Large\bf Limitations on separable measurements by convex optimization}

\author{Somshubhro Bandyopadhyay\thanks{
    Department of Physics and Center for Astroparticle Physics and
    Space Science, Bose Institute, India.} 
  \and
  Alessandro Cosentino\thanks{
    Institute for Quantum Computing, University of Waterloo, Canada.}
  \textsuperscript{ ,}\thanks{
    School of Computer Science, University of Waterloo, Canada.}
  \and
  Nathaniel Johnston\footnotemark[2]
  \textsuperscript{ ,}\thanks{
    Department of Combinatorics and Optimization, University of Waterloo,
    Canada.}\vspace{0.5mm}
  \and
  Vincent Russo \footnotemark[2]
  \textsuperscript{ ,}\footnotemark[3]
  \and
  John Watrous \footnotemark[2]
  \textsuperscript{ ,}\footnotemark[3]
  \textsuperscript{ ,}\thanks{Canadian Institute for
    Advanced Research, Toronto, Canada.}
  \and
  Nengkun Yu\footnotemark[2]
  \textsuperscript{ ,}\footnotemark[4]
  \textsuperscript{ ,}\thanks{Department of Mathematics and Statistics, 
    University of Guelph, Canada.}
}

\date{August 29, 2014\\[2mm](with minor corrections on January 4, 2021)}

\maketitle

\begin{abstract}
  We prove limitations on LOCC and separable measurements in bipartite
  state discrimination problems using techniques from convex optimization.
  Specific results that we prove include: an exact formula for the optimal 
  probability of correctly discriminating any set of either three or four Bell
  states via LOCC or separable measurements when the parties are given an
  ancillary partially entangled pair of qubits;
  an easily checkable characterization of when an unextendable product set is
  perfectly discriminated by separable measurements, along with the first known
  example of an unextendable product set that cannot be perfectly discriminated
  by separable measurements;
  and an optimal bound on the success probability for any LOCC or separable
  measurement for the recently proposed state discrimination problem of Yu,
  Duan, and Ying.
\end{abstract}

\section{Introduction}
\label{sec:intro}

The paradigm of \emph{local operations and classical communication}, or
\emph{LOCC} for short, is fundamental within the theory of quantum
information.
A protocol involving two or more individuals is said to be an 
\emph{LOCC protocol} when it may be implemented by means of classical
communication among the individuals, along with arbitrary quantum operations
performed locally.
This paradigm serves as a foundation through which properties of entanglement
may be studied, particularly those connected with the notion of entanglement as
a resource for information processing.

\subsubsection*{State discrimination problems}

One basic problem concerning the LOCC paradigm that has been studied in depth
regards the discrimination of sets of bipartite (or multipartite) states by
means of measurements that can be implemented by LOCC protocols.
In the most typically considered variant of this problem, one first specifies
an ensemble of states
\begin{equation}
  \bigl\{ (p_1,\rho_1),\ldots,(p_N,\rho_N)\bigr\},
\end{equation}
where $N$ is a positive integer, $(p_1,\ldots,p_N)$ is a probability vector,
and $\rho_1,\ldots,\rho_N$ are density operators representing quantum states
of systems shared between two separate individuals: Alice and Bob.
With respect to the probability vector $(p_1,\ldots,p_N)$, an index
$k\in\{1,\ldots,N\}$ is selected at random, and Alice and Bob are given the
quantum state $\rho_k$ for the selected index $k$.
Their goal is to determine the index $k$ of the given state~$\rho_k$
by means of an LOCC measurement.

In most prior works on this problem, the probability vector $(p_1,\ldots,p_N)$
has been taken to be $(1/N,\ldots,1/N)$, representing a uniform probability
distribution, and the states $\rho_1,\ldots,\rho_N$ have been taken to be pure
and orthogonal, so that a global measurement can trivially discriminate them
with certainty.
Many examples are known of specific choices of pure, orthogonal states
$\rho_1,\ldots,\rho_N$ for which a perfect discrimination is not possible
through LOCC measurements.
Some of these examples, along with other general results concerning this
problem, may be found in 
\cite{
  BandyopadhyayGK11,
  BandyopadhyayN2013,
  BandyopadhyayW2009,
  BennettDFMRSSW99,
  ChenL03,
  DuanFXY09,
  Fan04,
  GhoshKRS01,
  GhoshKRS04,
  HayashiMMOV06,
  HorodeckiSSH03,
  Nathanson05,
  Nathanson2013,
  WalgateH02,
  WalgateSHV00,
  Watrous05,
  YuDY12,
  YuDY2014}.

As perhaps the simplest example of an instance of this problem where a perfect
LOCC discrimination is not possible, one has that four Bell states,
\begin{equation} \label{eq:Bell-states}
  \begin{aligned}
    \ket{\phi_1} & = \frac{1}{\sqrt{2}}\ket{0}\ket{0} 
    + \frac{1}{\sqrt{2}}\ket{1}\ket{1}, & \qquad
    \ket{\phi_2} & = \frac{1}{\sqrt{2}}\ket{0}\ket{0} 
    - \frac{1}{\sqrt{2}}\ket{1}\ket{1},\\
    \ket{\phi_3} & = \frac{1}{\sqrt{2}}\ket{0}\ket{1} 
    + \frac{1}{\sqrt{2}}\ket{1}\ket{0}, & \qquad
    \ket{\phi_4} & = \frac{1}{\sqrt{2}}\ket{0}\ket{1} 
    - \frac{1}{\sqrt{2}}\ket{1}\ket{0},
  \end{aligned}
\end{equation}
cannot be perfectly discriminated by LOCC measurements \cite{GhoshKRS01}.
More precisely, if one takes $N = 4$, $\rho_k = \ket{\phi_k}\bra{\phi_k}$,
and $p_k = 1/4$ for each $k\in\{1,\ldots,4\}$ in the above problem, it holds
that the optimal probability with which Alice and Bob can correctly determine
the chosen index $k \in\{1,\ldots,4\}$, assuming that they are restricted to
local operations and classical communication, is $1/2$.
The fact that Alice and Bob can achieve a success probability of 1/2 is
straightforward: if they both measure their qubit with respect to the standard
basis, compare the results through classical communication, and answer
$k = 1$ if the measurements agree and $k=3$ if the measurements disagree, they
will be correct with probability 1/2.
The fact that they cannot achieve a probability of correctness larger than 1/2
follows from a result of Nathanson \cite{Nathanson05} stating that $N$
equiprobable, \emph{maximally entangled}, bipartite states having local
dimension $n$ can be discriminated correctly with probability at most $n/N$.
Consequently, even in the variant of this example in which $N=3$, 
$\rho_k = \ket{\phi_k}\bra{\phi_k}$, and $p_k = 1/3$ for each $k\in\{1,2,3\}$,
Alice and Bob cannot achieve a probability of correctness larger than 2/3,
which again is achievable through the simple protocol described above.
(More generally, if the states $\rho_1,\rho_2,\rho_3,\rho_4$ are given with
probabilities $p_1\geq p_2\geq p_3\geq p_4$, the optimal probability of a
correct discrimination by an LOCC measurement is $p_1 + p_2$
\cite{BandyopadhyayN2013}.)

Among the other known examples of collections of orthogonal pure states that
cannot be perfectly discriminated by LOCC protocols, the so-called \emph{domino
  state} example of \cite{BennettDFMRSSW99} is noteworthy.
In this example, the local dimension of the states is $3$, and one takes
$N = 9$, $p_1 = \cdots = p_9 = 1/9$, and
{\setlength{\arraycolsep}{1.2mm}%
\begin{equation} \label{eq:domino}
\begin{array}{llll}
  \multicolumn{4}{c}{\ket{\phi_1} = \ket{1}\ket{1},}\\[4mm]
  \ket{\phi_2} = \ket{0}\left(\frac{\ket{0} + \ket{1}}{\sqrt{2}}\right),
  & \ket{\phi_3} = \ket{2}\left(\frac{\ket{1} + \ket{2}}{\sqrt{2}}\right),
  & \ket{\phi_4} = \left(\frac{\ket{1} + \ket{2}}{\sqrt{2}}\right)\ket{0},
  & \ket{\phi_5} = \left(\frac{\ket{0} + \ket{1}}{\sqrt{2}}\right)\ket{2},
  \\[4mm]
  \ket{\phi_6} = \ket{0}\left(\frac{\ket{0} - \ket{1}}{\sqrt{2}}\right),
  & \ket{\phi_7} = \ket{2}\left(\frac{\ket{1} - \ket{2}}{\sqrt{2}}\right),
  & \ket{\phi_8} = \left(\frac{\ket{1} - \ket{2}}{\sqrt{2}}\right)\ket{0},
  & \ket{\phi_9} = \left(\frac{\ket{0} - \ket{1}}{\sqrt{2}}\right)\ket{2}.
\end{array}
\end{equation}
}%
A rather complicated argument demonstrates that this collection cannot be
discriminated with probability greater than $1 - \varepsilon$ for some choice
of a positive real number $\varepsilon$.
(A somewhat simplified proof appears in \cite{ChildsLMO13}, where this fact
is proved for $\varepsilon = 1.9 \times 10^{-8}$.)
The particular relevance of this example lies in the fact that all of these
states are product states, demonstrating that entanglement is not a requisite
for a set of orthogonal pure states to fail to be perfectly discriminated by
any LOCC measurement.
Complementary to this observation, a fundamental result of Walgate, Short,
Hardy, and Vedral \cite{WalgateSHV00} states that any two orthogonal pure
states (whether entangled or not) can always be perfectly discriminated by an
LOCC measurement.

\subsubsection*{Separable measurements}

The set of measurements that can be implemented through LOCC has an apparently
complex mathematical structure---no tractable characterization of this set is
known, representing a clear obstacle to a better understanding of the
limitations of LOCC measurements.
For example, given a collection of measurement operators $\{P_1,\ldots,P_N\}$
describing a measurement on a bipartite system, the determination of whether or
not this collection describes an LOCC measurement, or is closely approximated
by an LOCC measurement, is not known to be a computationally decidable
problem.

For this reason, the state discrimination problem described above is sometimes
considered for more tractable classes of measurements that approximate, in some
sense, the LOCC measurements.
The class of \emph{separable measurements} represents one commonly studied
approximation in this category.
Let us assume hereafter that $\X = \complex^n$ and $\Y = \complex^m$ are
complex Euclidean spaces (or, equivalently, finite-dimensional complex Hilbert
spaces) representing the local systems of Alice and Bob, respectively, in a
state discrimination problem.
A positive semidefinite operator $P\in\Pos(\X\otimes\Y)$ is said to be
\emph{separable} if it is possible to write
\begin{equation}
  P = \sum_{k = 1}^M Q_k \otimes R_k
\end{equation}
for some choice of a positive integer $M$ and positive semidefinite operators
$Q_1,\ldots,Q_M \in \Pos(\X)$ and $R_1,\ldots,R_M\in\Pos(\Y)$;
and a measurement $\{P_1,\ldots,P_N\}$ on $\X\otimes\Y$ is said to be a
separable measurement if it is the case that each measurement operator $P_k$ is
separable.
(Here, and elsewhere in the paper, $\Pos(\X)$, $\Pos(\Y)$, and
$\Pos(\X\otimes\Y)$ denote the sets of positive semidefinite operators acting
on $\X$, $\Y$, and $\X\otimes\Y$, respectively.)

It is straightforward to prove that every measurement that can be implemented
through LOCC is necessarily a separable measurement, from which it follows that
any limitation proved to hold for every separable measurement must also hold
for every LOCC measurement.
There are, however, separable measurements that cannot be implemented by LOCC
protocols; a measurement with respect to the orthonormal basis given by the
domino state example \eqref{eq:domino} is the archetypal example.

Many of the known results that establish limitations on LOCC measurements for
state discrimination tasks hold more generally for separable measurements, and
may be proved within this somewhat simpler setting.
For instance, the result of Nathanson mentioned earlier establishes that, in
the bipartite setting in which $n = m$ (i.e., Alice and Bob's local systems are
represented by $\X = \Y = \complex^n$), one has that $N$ equiprobable, maximally
entangled states cannot be discriminated with success probability exceeding
$n/N$ by any separable measurement.

\subsubsection*{PPT measurements}

Another class of measurements, representing a further relaxation of the LOCC
condition, is the class of PPT measurements.
A positive semidefinite operator $P\in \Pos(\X\otimes\Y)$ is a \emph{PPT}
(short for \emph{positive partial transpose}) operator if it holds that
\begin{equation}
  \pt_{\negsmallspace\X}(P) \in \Pos(\X\otimes\Y),
\end{equation}
where $\pt_{\negsmallspace\X}:\Lin(\X\otimes\Y)\rightarrow\Lin(\X\otimes\Y)$ is
the linear mapping representing partial transposition with respect to the
standard basis $\{\ket{0},\ldots,\ket{n-1}\}$ of $\X$.
A measurement is a \emph{PPT measurement} if it is represented by a collection
of PPT measurement operators $\{P_1,\ldots,P_N\}$.
Every separable operator is a PPT operator, so that every separable measurement
(and therefore every LOCC measurement) is a PPT measurement as well. 

The primary appeal of the set of PPT measurements is its mathematical
simplicity.
In particular, the PPT condition is represented by linear and positive
semidefinite constraints, which allows for an optimization over the collection
of PPT measurements to be represented by a semidefinite program.
By the duality theory of semidefinite programs, one may obtain upper bounds on
the success probability of any PPT measurement (and therefore any LOCC or
separable measurement) for a given state discrimination problem;
this may be done by simply exhibiting a feasible solution to the dual problem 
of the semidefinite program representing the measurement optimization for this
set of states.
The downside of this approach is that the set of PPT measurements is a
coarse approximation to the set of LOCC measurements, so the method will fail
to prove strong impossibility results for LOCC measurements in many cases.

The approach described above, in which PPT measurements are represented by
semidefinite programs, was taken in \cite{Cosentino13}.
There, it was shown that the state discrimination problem of Yu, Duan, and
Ying \cite{YuDY12}, to discriminate the four maximally entangled states
\begin{equation}
  \label{eq:ydy_states}
  \begin{aligned}
    \ket{\phi_1} & =
    \frac{1}{2} \ket{0}\ket{0} + 
    \frac{1}{2} \ket{1}\ket{1} + 
    \frac{1}{2} \ket{2}\ket{2} + 
    \frac{1}{2} \ket{3}\ket{3} \\[2mm]
    \ket{\phi_2} & =
    \frac{1}{2} \ket{0}\ket{3} + 
    \frac{1}{2} \ket{1}\ket{2} + 
    \frac{1}{2} \ket{2}\ket{1} + 
    \frac{1}{2} \ket{3}\ket{0} \\[2mm]
    \ket{\phi_3} & =
    \frac{1}{2} \ket{0}\ket{3} +
    \frac{1}{2} \ket{1}\ket{2} - 
    \frac{1}{2} \ket{2}\ket{1} - 
    \frac{1}{2} \ket{3}\ket{0} \\[2mm]
    \ket{\phi_4} & =
    \frac{1}{2} \ket{0}\ket{1} + 
    \frac{1}{2} \ket{1}\ket{0} -
    \frac{1}{2} \ket{2}\ket{3} - 
    \frac{1}{2} \ket{3}\ket{2}
  \end{aligned}
\end{equation}
by means of a PPT measurement, assuming the states are given with equal
probability, is at most 7/8.
(Yu, Duan, and Ying proved that this set cannot be perfectly discriminated by
any PPT measurement through a different argument, giving the first example
where $N=n$ maximally entangled states having local dimension $n$ cannot be
perfectly discriminated by LOCC measurements.
Examples of sets of $N<n$ states that cannot be discriminated without error
by PPT measurements were later given in \cite{CosentinoR13}.)

\subsubsection*{Overview of results}

This paper proves several new results concerning state discrimination by LOCC,
separable, and PPT measurements using techniques based on convex optimization.
Our results are primarily focused on separable measurements, and are mostly
based on the paradigm of \emph{cone programming}, which is a generalization of
linear programming and semidefinite programming that allows for optimizations
over general closed, convex cones.
We also obtain new results based on linear programming and semidefinite
programming.
The notion of \emph{duality}, shared by linear programs, semidefinite programs,
and cone programs, plays a central role in our results.
The following specific results are among those we prove:
\begin{mylist}{\parindent}
\item[$\bullet$]
  We obtain an exact formula for the optimal probability of correctly
  discriminating any set of either three or four Bell states via separable
  measurements, when the parties are given an ancillary partially entangled
  pair of qubits.
  In particular, it is proved that this ancillary pair of qubits must be
  maximally entangled in order for three Bell states to be perfectly
  discriminated by separable (or LOCC) measurements, which answers an open
  question of \cite{YuDY2014}.

\item[$\bullet$]
  We provide an easily checkable characterization of when an unextendable 
  product set is perfectly discriminated by separable measurements, and we 
  use this characterization to present an example of an unextendable product
  set in $\X\otimes\Y$, for $\X  = \Y = \complex^4$, that is not perfectly 
  discriminated by separable measurements.
  This resolves an open question raised in \cite{DuanFXY09}.
  We also show that every unextendable product set together with one extra pure
  state orthogonal to every member of the unextendable product set is not
  perfectly discriminated by separable measurements.

\item[$\bullet$]
  It is proved that the maximum success probability for any separable
  measurement in the state discrimination problem of Yu, Duan, and Ying
  specified above is 3/4.
  This bound is easily seen to be achievable by an LOCC measurement, implying
  that it is the optimal success probability of an LOCC measurement for this
  problem.
  The upper-bound is closely connected to the positive maps of Breuer and Hall
  \cite{Breuer06,Hall06}.
\end{mylist}

\section{A cone program for optimizing over separable measurements}
\label{sec:cone}

A \emph{cone program} (also known as a \emph{conic program}) expresses the
maximization of a linear function over the intersection of an affine subspace
and a closed convex cone in a finite-dimensional real inner product
space \cite{BV04}.
Linear programming and semidefinite programming are special cases of cone
programming: in linear programming, the closed convex cone over which the
optimization occurs is the positive orthant in $\real^n$, while in semidefinite
programming the optimization is over the cone $\Pos(\complex^n)$ of positive
semidefinite operators on $\complex^n$.
In the case of semidefinite programming, the finite-dimensional real inner
product space is the real vector space $\Herm(\complex^n)$ of Hermitian
operators on $\complex^n$, equipped with the Hilbert-Schmidt inner product:
$\ip{X}{Y} = \tr(XY)$.
One may also consider semidefinite programming over real positive semidefinite
operators.

\subsubsection*{Cone programming}

For the purposes of the present paper, it is sufficient to consider only cone
programs defined over spaces of Hermitian operators (with the Hilbert-Schmidt
inner product).
In particular, let $\Z$ and $\W$ be complex Euclidean spaces, let $\Herm(\Z)$
and $\Herm(\W)$ denote the sets of Hermitian operators acting on $\Z$ and $\W$,
respectively, and let $\K\subseteq\Herm(\Z)$ be a closed, convex cone.
For any choice of a linear map $\Phi:\Herm(\Z)\rightarrow\Herm(\W)$ and
Hermitian operators $A\in\Herm(\Z)$ and $B\in\Herm(\W)$, one has a cone
program, which is represented by a pair of optimization problems:
\begin{center}
  \begin{minipage}{2in}
    \centerline{\underline{Primal problem}}\vspace{-7mm}
    \begin{align*}
      \text{maximize:}\quad & \ip{A}{X}\\
      \text{subject to:}\quad & \Phi(X) = B,\\
      & X \in \K.
    \end{align*}
  \end{minipage}
  \hspace*{1.5cm}
  \begin{minipage}{2.4in}
    \centerline{\underline{Dual problem}}\vspace{-7mm}
    \begin{align*}
      \text{minimize:}\quad & \ip{B}{Y}\\
      \text{subject to:}\quad & \Phi^{\ast}(Y) - A \in \K^{\ast},\\
      & Y\in\Herm(\W).
    \end{align*}
  \end{minipage}
\end{center}
Here, $\K^{\ast}$ denotes the \emph{dual cone} to $\K$, defined as
\begin{equation}
  \K^{\ast} = \{Y\in\Herm(\Z)\,:\,\ip{X}{Y} \geq 0\;\:\text{for all $X\in\K$}\}
\end{equation}
and $\Phi^{\ast}:\Herm(\W)\rightarrow\Herm(\Z)$ is the adjoint mapping to
$\Phi$, which is uniquely determined by the equation
$\ip{Y}{\Phi(X)} = \ip{\Phi^{\ast}(Y)}{X}$ holding for all $X\in\Herm(\Z)$ and
$Y\in\Herm(\W)$.

With the optimization problems above in mind, one defines the 
\emph{feasible sets} $\A$ and $\B$ of the primal and dual problems as
\begin{equation}
  \A = \bigl\{ X \in \K : \Phi(X) = B\bigr\} 
  \qquad \text{and} \qquad 
  \B = \bigl\{ Y \in \Herm(\W) : \Phi^{\ast}(Y) - A \in \K^{\ast} \bigr\}.
\end{equation}
One says that the associated cone program is \emph{primal feasible} if 
$\A \neq \emptyset$, and is \emph{dual feasible} if $\B \neq \emptyset$. 
The function $X \mapsto \ip{A}{X}$ from $\Herm(\Z)$ to $\real$ is called the 
\emph{primal objective function}, and the function $Y \mapsto \ip{B}{Y}$ from 
$\Herm(\W)$ to $\real$ is called the \emph{dual objective function}.
The \emph{optimal values} associated with the primal and dual problems are
defined as
\begin{equation}
  \label{eq:alpha-and-beta}
  \alpha = \sup \bigl\{ \ip{A}{X} : X \in \A \bigr\} 
  \qquad \text{and} \qquad 
  \beta = \inf \bigl\{ \ip{B}{Y} : Y \in \B \bigr\},
\end{equation}
respectively.
(It is conventional to interpret that $\alpha = -\infty$ when $\A = \emptyset$
and $\beta = \infty$ when $\B = \emptyset$.)
The property of \emph{weak duality}, which holds for all cone programs, is
that the primal optimum can never exceed the dual optimum.

\begin{prop}[Weak duality for cone programs]
  For any choice of complex Euclidean spaces $\Z$ and $\W$, a closed, convex
  cone $\K\subseteq\Herm(\Z)$, Hermitian operators $A\in\Herm(\Z)$ and
  $B\in\Herm(\W)$, and a linear map $\Phi:\Herm(\Z)\rightarrow\Herm(\W)$, it
  holds that $\alpha \leq \beta$, for $\alpha$ and $\beta$ as defined in
  \eqref{eq:alpha-and-beta}.
\end{prop}

\begin{proof}
The proposition is trivial in case $\A = \emptyset$ (which implies that 
$\alpha = -\infty)$ or $\B = \emptyset$ (which implies that $\beta = \infty$), 
so we will restrict our attention to the case that both $\A$ and $\B$ are 
nonempty. 
For any choice of $X \in \A$ and $Y \in \B$, one must have $X \in \K$ and
$\Phi^{\ast}(Y) - A \in \K^{\ast}$, and therefore
$\ip{\Phi^{\ast}(Y) - A}{X} \geq 0$.
It follows that
\begin{equation}
  \ip{A}{X} = \ip{\Phi^{\ast}(Y)}{X} - \ip{\Phi^{\ast}(Y) - A}{X}
  \leq \ip{Y}{\Phi(X)} = \ip{B}{Y}.
\end{equation}
Taking the supremum over all $X \in \A$ and the infimum over all 
$Y \in \B$ establishes that $\alpha \leq \beta$.
\end{proof}

Weak duality implies that every dual feasible operator $Y \in \B$ provides an
upper bound of $\ip{B}{Y}$ on the value $\ip{A}{X}$ that is achievable over all
choices of a primal feasible $X \in \A$, and likewise every primal feasible 
operator $X \in \A$ provides a lower bound of $\ip{A}{X}$ on the value 
$\ip{B}{Y}$
that is achievable over all choices of a dual feasible solution $Y \in \B$. 
In other words, it holds that
$\ip{A}{X} \leq \alpha \leq \beta \leq \ip{B}{Y}$,
for every $X \in \A$ and $Y \in \B$. 

There are simple conditions under which the primal and dual optimal values will
in fact be equal, which is a situation known as \emph{strong duality}.
Although the cone programs considered in this paper do indeed possess this
stronger notion of duality, it is not needed for any of our results.

\subsubsection*{Optimizing over separable measurements}

Let us now return to the state discrimination problem.
Let $\X = \complex^n$ and $\Y=\complex^m$ be complex Euclidean spaces
corresponding to quantum systems held by Alice and Bob, respectively, and let
$\{\rho_1, \ldots, \rho_N\} \subset \Density(\X \otimes \Y)$ be a set
of density operators, representing quantum states of Alice and Bob's shared
systems.
Alice and Bob are given a state $\rho_k$, for some index $k\in\{1,\ldots,N\}$
drawn according to a probability distribution $(p_1,\ldots,p_N)$, as was 
described above in the introduction.
Their goal is to maximize the probability that they correctly identify the
chosen index $k$, assuming that they have complete knowledge of the set
$\{\rho_1,\ldots,\rho_N\}$ and the probability distribution
$(p_1,\ldots,p_N)$.
Our focus is on the situation in which they do this by means of a separable
measurement $\{P_1,\ldots,P_N\}$.

Hereafter, let us denote the set of all separable operators acting on the space
$\X\otimes\Y$ by $\Sep(\X : \Y)$.
One may observe that $\Sep(\X : \Y)$ is a closed, convex cone, which will allow
an optimization over separable measurements $\{P_1,\ldots,P_N\}$ to be
expressed as a cone program.
The dual cone to $\Sep(\X : \Y)$, which is commonly known as the set of 
\emph{block-positive operators}, is defined as
\begin{equation}
  \BPos(\X:\Y) = \bigl\{
  H\in\Herm(\X\otimes\Y)\,:\,
  \ip{P}{H} \geq 0 \;\text{for every $P \in \Sep(\X:\Y)$}\bigr\}.
\end{equation}
There are several equivalent characterizations of this set.
For instance, one has
\begin{equation}
  \BPos(\X:\Y) = \bigl\{
  H\in\Herm(\X\otimes\Y)\,:\,
  (\I_{\X} \otimes y^{\ast}) H (\I_{\X} \otimes y)\in\Pos(\X)
  \;\text{for every $y\in\Y$}\bigr\}.
\end{equation}
Alternatively, block-positive operators can be characterized as representations 
of \emph{positive} linear maps, via the Choi representation.
That is, for any linear map $\Phi:\Lin(\Y)\rightarrow\Lin(\X)$
mapping arbitrary linear operators on $\Y$ to linear operators on $\X$, the
following two properties are equivalent:
\begin{itemize}
\item[(a)] For every positive semidefinite operator $Y \in \Pos(\Y)$, it 
holds that $\Phi(Y) \in \Pos(\X)$. 
\item[(b)] The Choi operator
\begin{equation}
  J(\Phi) = \sum_{0\leq j,k< m} \Phi\bigl(\ket{j}\bra{k}\bigr) 
  \otimes \ket{j}\bra{k}
\end{equation}
of $\Phi$ satisfies $J(\Phi) \in \BPos(\X:\Y)$.
\end{itemize}

We now observe that the following cone program represents the optimal value
of a correct state discrimination in the setting under consideration.
The primal problem is as follows:
\begin{equation}
  \label{eq:primal-problem}
  \begin{split}
    \text{maximize:} \quad & 
    p_1 \ip{\rho_1}{P_1} + \cdots + p_N \ip{\rho_N}{P_N}\\
    \text{subject to:} \quad & P_1 + \cdots + P_N = \I_{\X\otimes\Y}\\
    & P_k\in \Sep(\X:\Y)\quad(\text{for each}\;k = 1,\ldots,N),
  \end{split}
\end{equation}
and the dual problem is as follows:
\begin{equation}
  \label{eq:dual-problem}
  \begin{split}
    \text{minimize:} \quad & \tr(H)\\
    \text{subject to:} \quad & H-p_k\rho_k\in\BPos(\X:\Y)
    \quad(\text{for each}\;k = 1,\ldots,N)\\
    \quad & H \in \Herm(\X\otimes\Y).
  \end{split}
\end{equation}

If one is to formally specify this problem according to the general form for
cone programs presented above, the operators $P_1,\ldots,P_N$ may be
represented as a block matrix of the form
\begin{equation}
  X = \begin{pmatrix}
    P_1 & \cdots & \cdot \\
    \vdots & \ddots & \vdots\\
    \cdot & \cdots & P_N
  \end{pmatrix} \in \Herm((\X\otimes\Y) \oplus \cdots \oplus (\X\otimes\Y))
\end{equation}
with the off-diagonal blocks being left unspecified.
The cone $\K$ is taken to be the cone of operators of this form for which each
$P_k$ is separable, and the mapping $\Phi$ and operators $A$ and $B$ are
chosen in the natural way:
\begin{equation}
  A = \begin{pmatrix}
    p_1 \rho_1 & \cdots & 0 \\
    \vdots & \ddots & \vdots\\
    0 & \cdots & p_N\rho_N
  \end{pmatrix},
  \quad
  B = \I, \quad\text{and}\quad
  \Phi\begin{pmatrix}
  P_1 & \cdots & \cdot \\
  \vdots & \ddots & \vdots\\
  \cdot & \cdots & P_N
  \end{pmatrix}
  = P_1+\cdots+P_N.
\end{equation}
One may verify that the dual problem is as claimed by a straightforward
computation.

By weak duality for cone programs, an upper bound on the probability with which
a separable measurement can discriminate the states $\rho_1,\ldots,\rho_N$
given with probabilities $p_1,\ldots,p_N$ is given by every dual feasible
solution to this cone program:
for any Hermitian operator $H\in\Herm(\X\otimes\Y)$ for which $H-p_k\rho_k$ is
block positive for every $k = 1,\ldots,N$, the probability of a correct
discrimination is upper-bounded by $\tr(H)$.
\section{The entanglement cost of discriminating Bell states}
\label{sec:bell}

As explained in the introduction, three Bell states given with uniform
probabilities can be discriminated by separable measurements with success
probability at most 2/3, while four can be discriminated with success
probability at most 1/2.
These bounds can be obtained by a fairly trivial selection of LOCC
measurements, and can be shown to hold even for PPT measurements.

In this section, we study state discrimination problems for sets of three or
four Bell states, by LOCC, separable, and PPT measurements, with the assistance
of an entangled pair of qubits.
In particular, we will assume that Alice and Bob aim to discriminate a set of
Bell states given that they share the additional resource state
\begin{equation}
  \label{eq:tau_eps}
  \ket{\tau_{\eps}} = \sqrt{\frac{1 + \eps}{2}}\,\ket{0}\ket{0} + 
  \sqrt{\frac{1 - \eps}{2}}\,\ket{1}\ket{1},
\end{equation}
for some choice of $\eps \in [0,1]$.
The parameter $\eps$ quantifies the amount of entanglement in the state
$\ket{\tau_{\eps}}$.
Up to local unitaries, this family of states represents every pure state of two
qubits.

The entanglement cost of quantum operations and measurements, within the
paradigm of LOCC, has been considered previously.
For instance, \cite{Cohen2008} studied the entanglement cost of perfectly
discriminating elements of unextendable product sets by LOCC measurements,
\cite{BandyopadhyayBKW2009} and \cite{BandyopadhyayRW2010} considered the
entanglement cost of measurements and established lower bounds on the amount of
entanglement necessary for distinguishing complete orthonormal bases of two
qubits, and \cite{YuDY2014} considered the entanglement cost of state
discrimination problems by PPT and separable measurements.

Using the cone programming method discussed in the previous section, we obtain
exact expressions for the optimal probability with which any set of three or
four Bell states can be discriminated with the assistance of the state
\eqref{eq:tau_eps} by separable measurements (which match the probabilities
obtained by LOCC measurements in all cases).
This answers an open question raised in \cite{YuDY2014}.

\subsubsection*{Discriminating three Bell states}

Notice that the state $\ket{\tau_{1}} = \ket{0}\ket{0}$ is a product state and
it does not aid the two parties in discriminating any set of Bell states,
so the probability of success for $\eps = 1$ is still at most $2/3$ for a set
of three Bell states.
If $\varepsilon = 0$, then Alice and Bob can use teleportation to perfectly
discriminate all four Bell states perfectly by LOCC measurements, and therefore
the same is true for any three Bell states.
It was proved in \cite{YuDY2014} that PPT measurements can perfectly
discriminate any set of three Bell states using the resource state
\eqref{eq:tau_eps} if and only if $\eps \leq 1/3$.

Here we show that a maximally entangled state ($\eps = 0$) is required to 
perfectly discriminate any set of three Bell states using separable
measurements, and more generally we obtain an expression for the optimal
probability of a correct discrimination for all values of $\varepsilon$.
Because the permutations of Bell states induced by local unitaries is
transitive, there is no loss of generality in fixing the three Bell states to
be discriminated to be $\ket{\phi_1}$, $\ket{\phi_2}$, and $\ket{\phi_3}$
(as defined in \eqref{eq:Bell-states}).

\begin{theorem}
  \label{thm:three-bell}
  Let $\X_1 = \X_2 = \Y_1 = \Y_2 = \complex^2$, define 
  $\X = \X_1 \otimes \X_2$ and $\Y = \Y_1 \otimes \Y_2$, and let 
  $\eps \in [0,1]$ be chosen arbitrarily. 
  For any separable measurement $\{P_1,P_2,P_3\}\subset\Sep(\X:\Y)$, the
  success probability of correctly discriminating the states corresponding to
  the set
  \begin{equation}
    \label{eq:set-three-bells}
    \bigl\{ \ket{\phi_{1}} \otimes \ket{\tau_{\eps}},\; 
    \ket{\phi_{2}} \otimes \ket{\tau_{\eps}},\;
    \ket{\phi_{3}} \otimes \ket{\tau_{\eps}} \bigr\}
    \subset (\X_1\otimes\Y_1)\otimes(\X_2\otimes\Y_2),
  \end{equation}
  assuming a uniform distribution $p_1 = p_2 = p_3 = 1/3$, is at most
  \begin{equation}
  \label{eq:probability-three-bell}
    \frac{1}{3}\left(2 + \sqrt{1 - \eps^{2}}\right).
  \end{equation}
\end{theorem}

To prove this theorem, we require the following lemma.
The lemma introduces a family of positive maps that, to our knowledge, has not
previously appeared in the literature.

\begin{lemma}
  \label{lemma:3Bell}
  Define a linear mapping
  $\Xi_{t}: \Lin(\complex^2 \oplus \complex^2)\rightarrow
  \Lin(\complex^2 \oplus \complex^2)$ as
  \begin{equation}
    \Xi_t\begin{pmatrix}
    A & B\\
    C & D
    \end{pmatrix}
    = \begin{pmatrix}
      \Psi_t(D) + \Phi(D) &
      \Psi_t(B) + \Phi(C)\\[2mm]
      \Psi_t(C) + \Phi(B) &
      \Psi_t(A) + \Phi(A)
    \end{pmatrix}
  \end{equation}
  for every $t\in(0,\infty)$ and $A,B,C,D\in\Lin(\complex^2)$, where
  $\Psi_t:\Lin(\complex^2)\rightarrow\Lin(\complex^2)$
  is defined as
  \begin{equation}
    \Psi_t
    \begin{pmatrix}
      \alpha & \beta \\
      \gamma & \delta
    \end{pmatrix}
    = 
    \begin{pmatrix}
      t \alpha & \beta \\
      \gamma & t^{-1} \delta
    \end{pmatrix}
  \end{equation}
  and $\Phi:\Lin(\complex^2)\rightarrow\Lin(\complex^2)$ is defined as
  \begin{equation}
    \Phi\begin{pmatrix}
    \alpha & \beta \\
    \gamma & \delta
    \end{pmatrix}
    = \begin{pmatrix}
      \delta & -\beta\\
      -\gamma & \alpha
    \end{pmatrix},
  \end{equation}
  for every $\alpha,\beta,\gamma,\delta\in\complex$.
  It holds that $\Xi_t$ is a positive map for all $t\in (0,\infty)$.
\end{lemma}

\begin{proof}
  It will first be proved that $\Xi_1$ is positive.
  For every vector
  \begin{equation}
    u = \begin{pmatrix}
      \alpha\\ \beta
    \end{pmatrix}
  \end{equation}
  in $\complex^2$, define a matrix
  \begin{equation}
    M_u = \begin{pmatrix}
      \overline\alpha & \overline\beta\\[1mm]
      -\beta & \alpha
    \end{pmatrix}.
  \end{equation}
  Straightforward computations reveal that
  \begin{equation}
    M_u^{\ast} M_v = u v^{\ast} + \Phi(v u^{\ast})
    \qquad\text{and}\qquad
    M_u^{\ast} M_u = \norm{u}^2\tinyspace \I
  \end{equation}
  for all $u,v\in\complex^2$.
  It follows that
  \begin{equation}
    \Xi_1 \begin{pmatrix}
      u u^{\ast} & u v^{\ast}\\
      v u^{\ast} & v v^{\ast}
    \end{pmatrix}
    = \begin{pmatrix}
      v v^{\ast} + \Phi(v v^{\ast}) & 
      u v^{\ast} + \Phi(v u^{\ast}) \\
      v u^{\ast} + \Phi(u v^{\ast}) &
      u u^{\ast} + \Phi(u u^{\ast}) 
    \end{pmatrix}
    = \begin{pmatrix}
      \norm{v}^2 \I & M_u^{\ast} M_v\\[1mm]
      M_v^{\ast} M_u & \norm{u}^2 \I
    \end{pmatrix},
  \end{equation}
  which is positive semidefinite by virtue of the fact that
  $\norm{M_u^{\ast} M_v}\leq\norm{M_u}\norm{M_v} = \norm{u} \norm{v}$.
  As every element of $\Pos(\complex^2\oplus\complex^2)$ can be written as a
  positive linear combination of matrices of the form
  \begin{equation}
    \begin{pmatrix}
      u u^{\ast} & u v^{\ast}\\
      v u^{\ast} & v v^{\ast}
    \end{pmatrix},
  \end{equation}
  ranging over all vectors $u, v \in \complex^2$, it follows that $\Xi_1$ is a
  positive map.

  For the general case, observe first that the mapping $\Psi_s$ may be
  expressed using the Hadamard (or entry-wise) product as
  \begin{equation}
    \Psi_s
    \begin{pmatrix}
      \alpha & \beta \\
      \gamma & \delta
    \end{pmatrix}
    = 
    \begin{pmatrix}
      s \alpha & \beta \\
      \gamma & s^{-1} \delta
    \end{pmatrix}
    =
    \begin{pmatrix}
      s & 1\\
      1 & s^{-1}
    \end{pmatrix} \circ
    \begin{pmatrix}
      \alpha & \beta \\
      \gamma & \delta
    \end{pmatrix}
  \end{equation}
  for every positive real number $s\in(0,\infty)$.
  The matrix
  \begin{equation}
    \begin{pmatrix}
      s & 1\\
      1 & s^{-1}
    \end{pmatrix}
  \end{equation}
  is positive semidefinite, from which it follows (by the Schur product
  theorem) that $\Psi_s$ is a completely positive map.
  (See, for instance, Theorem 3.7 of \cite{Paulsen02}.)
  Also note that $\Phi = \Psi_s \Phi \Psi_s$ for every $s\in (0,\infty)$, which
  implies that
  \begin{equation}
    \Xi_t = \bigl(\I_{\Lin(\complex^2)} \otimes \Psi_s\bigr) \Xi_1
    \bigl(\I_{\Lin(\complex^2)} \otimes \Psi_s\bigr)
  \end{equation}
  for $s = \sqrt{t}$.
  This shows that $\Xi_t$ is a composition of positive maps for every positive
  real number~$t$, and is therefore positive.
\end{proof}

\begin{proof}[Proof of Theorem \ref{thm:three-bell}]
  For the cases that $\eps = 0$ and $\varepsilon = 1$, the theorem is known,
  as was discussed previously, so it will be assumed that $\eps \in (0,1)$.
  Define a Hermitian operator
  \begin{equation}
    H_{\eps} = \frac{1}{3}\left[\frac{\I_{\X_1\otimes\Y_1}\otimes
        \tau_{\eps}}{2} + \sqrt{1 - \eps^{2}} \, \phi_{4} \otimes
      \pt_{\negsmallspace\X_2}(\phi_{4}) \right],
  \end{equation}
  where $\tau_{\eps} = \ket{\tau_{\eps}}\bra{\tau_{\eps}}$,
  $\phi_{4} = \ket{\phi_{4}}\bra{\phi_{4}}$,
  and $\pt_{\negsmallspace\X_2}$ denotes partial transposition with respect to
  the standard basis of $\X_2$.
  It holds that
  \begin{equation}
    \tr(H_{\eps}) = \frac{1}{3}\left(2 + \sqrt{1 - \eps^{2}}\right),
  \end{equation}
  so to complete the proof it suffices to prove that $H_{\varepsilon}$ is a
  feasible solution to the dual problem \eqref{eq:dual-problem} for the cone
  program corresponding to the state discrimination problem being considered.

  In order to be more precise about the task at hand, it is helpful to define a
  unitary operator $W$, mapping $\X_1\otimes\X_2\otimes\Y_1\otimes\Y_2$ to
  $\X_1\otimes\Y_1\otimes\X_2\otimes\Y_2$, that corresponds to swapping the
  second and third subsystems:
  \begin{equation}
    \label{eq:swap}
    W(x_{1}\otimes x_{2}\otimes y_{1}\otimes y_{2}) =
    x_{1}\otimes y_{1}\otimes x_{2}\otimes y_{2},
  \end{equation}
  for all vectors 
  $x_{1}\in\X_{1}$, $x_{2}\in\X_{2}$, $y_{1}\in\Y_{1}$, $y_{2}\in\Y_{2}$.
  We are concerned with the separability of measurement operators with respect
  to the bipartition between $\X_1\otimes\X_2$ and $\Y_1\otimes\Y_2$, so the
  dual feasibility of $H_{\varepsilon}$ requires that the operators defined as
  \begin{equation}
    Q_{k, \eps} = W^{\ast} \left( H_{\eps} - \frac{1}{3}\phi_{k} \otimes 
    \tau_{\eps} \right) W \in \Herm(\X\otimes\Y)
  \end{equation}
  be contained in $\BPos(\X:\Y)$ for $k = 1,2,3$.

  Let $\Lambda_{k, \eps}: \Lin(\Y) \rightarrow \Lin(\X)$ be the unique linear
  map whose Choi representation satisfies 
  $J(\Lambda_{k,\varepsilon}) = Q_{k, \eps}$ for each $k = 1,2,3$.
  As discussed in Section~\ref{sec:cone}, the block positivity of $Q_{k, \eps}$
  is equivalent to the positivity of $\Lambda_{k, \eps}$.
  Consider first the case $k = 1$ and let
  \begin{equation}
    t = \sqrt{\frac{1+\eps}{1-\eps}}.
  \end{equation}
  A calculation reveals that
  \begin{equation}
    \Lambda_{1, \eps}(Y) = \frac{\sqrt{1 - \eps^{2}}}{12}
    \left(\sigma_{3} \otimes \I_{\X_2}\right)
    \Xi_{t}(Y)
    \left(\sigma_{3} \otimes \I_{\X_2}\right),
  \end{equation}
  where $\Xi_{t}:\Lin(\Y)\rightarrow\Lin(\X)$ is the map defined in
  Lemma~\ref{lemma:3Bell} and (in general)
  \begin{equation} \label{eq:Pauli-operators}
  \begin{array}{llll}
    \sigma_{0} = \begin{pmatrix} 1 & 0 \\ 0 & 1 \end{pmatrix}, & 
    \sigma_{1} = \begin{pmatrix} 0 & 1 \\ 1 & 0 \end{pmatrix}, & 
    \sigma_{2} = \begin{pmatrix} 0 & -i \\ i & 0 \end{pmatrix}, & 
    \sigma_{3} = \begin{pmatrix} 1 & 0 \\ 0 & -1 \end{pmatrix}
  \end{array}
  \end{equation}
  denote the Pauli operators.
  As $\eps \in (0,1)$, it holds that $t\in (0,\infty)$, and therefore
  Lemma~\ref{lemma:3Bell} implies that $\Xi_{t}(Y) \in \Pos(\X)$ for every
  $Y \in \Pos(Y)$. 
  As we are simply conjugating $\Xi_{t}(Y)$ by a unitary and scaling it 
  by a positive real factor, we also have that 
  $\Lambda_{1, \eps}(Y) \in \Pos(\X)$, for any $Y \in \Pos(Y)$, which in turn
  implies that $Q_{1, \eps} \in \BPos(\X:\Y)$.

  For the case of $k=2$ and $k = 3$, first define 
  $U, V \in \Unitary(\complex^{2})$ as follows:
  \begin{equation}
    U = \begin{pmatrix}
      1 & 0\\
      0 & i
    \end{pmatrix}
    \quad \mbox{and} \quad
    V = \frac{1}{\sqrt{2}}\begin{pmatrix}
      1 & i\\
      i & 1
    \end{pmatrix}.
  \end{equation}
  These operators transform $\phi_{1} = \ket{\phi_1}\bra{\phi_1}$ into 
  $\phi_{2} = \ket{\phi_2}\bra{\phi_2}$ and $\phi_{3} =
  \ket{\phi_3}\bra{\phi_3}$, respectively, and leave $\phi_{4}$ unchanged, in
  the following sense:
  \begin{equation}
    \begin{aligned}
      (U^{\ast}\otimes U^{\ast}) \phi_{1} (U\otimes U) = \phi_{2},\\
      (V^{\ast}\otimes V^{\ast}) \phi_{1} (V\otimes V) = \phi_{3},\\
      (U^{\ast}\otimes U^{\ast}) \phi_{4} (U\otimes U) = \phi_{4},\\
      (V^{\ast}\otimes V^{\ast}) \phi_{4} (V\otimes V) = \phi_{4}.
    \end{aligned}
  \end{equation}
  Therefore the following equations hold:
  \begin{equation}
    \begin{aligned}
      Q_{2,\eps} &= \left(U^{\ast}\otimes\I \otimes U^{\ast}\otimes\I\right) 
      Q_{1,\eps} \left(U\otimes\I \otimes U\otimes\I\right),  \\
      Q_{3,\eps} &= \left(V^{\ast}\otimes\I \otimes V^{\ast}\otimes\I\right) 
      Q_{1,\eps} \left(V\otimes\I \otimes V\otimes\I\right).
    \end{aligned}
  \end{equation}
  It follows that $Q_{2,\eps}\in\BPos(\X:\Y)$ and $Q_{3,\eps}\in\BPos(\X:\Y)$,
  which completes the proof.
\end{proof}

\begin{remark}
  The upper bound obtained in Theorem \ref{thm:three-bell} is achievable by an
  LOCC measurement, as it is the probability obtained by using the resource
  state $\ket{\tau_\varepsilon}$ to teleport the given Bell state from one
  player to the other, followed by an optimal local measurement to discriminate
  the resulting states.
\end{remark}

\subsubsection*{Discriminating four Bell states}

It is known that, for the perfect LOCC discrimination of all four Bell states
using an auxiliary entangled state $\ket{\tau_{\varepsilon}}$ as above,
one requires that $\varepsilon = 0$ (i.e., a maximally entangled pair of qubits
is required).
This fact follows from the method of \cite{HorodeckiSSH03}, for instance.
Here we prove a more precise bound on the optimal probability of a correct
discrimination, for every choice of $\varepsilon\in[0,1]$, along similar lines
to the bound on three Bell states provided by Theorem~\ref{thm:three-bell}.
In the present case, in which all four Bell states are considered, the result
is somewhat easier: one obtains an upper bound for PPT measurements that
matches a bound that can be obtained by an LOCC measurement,
implying that LOCC, separable, and PPT measurements are equivalent for this
discrimination problem.

\begin{theorem}
  \label{thm:four-bell}
  Let $\X_1 = \X_2 = \Y_1 = \Y_2 = \complex^2$, define
  $\X = \X_1 \otimes \X_2$ and $\Y = \Y_1 \otimes \Y_2$, and let
  $\varepsilon\in [0,1]$.
  For any PPT measurement $\{P_1,P_2,P_3,P_4\}\subset\PPT(\X:\Y)$, the success
  probability of discriminating the states corresponding to the set
  \begin{equation}
    \label{eq:set-four-bells}
    \left\{ \ket{\phi_{1}} \otimes \ket{\tau_{\eps}},\; 
    \ket{\phi_{2}} \otimes \ket{\tau_{\eps}},\;
    \ket{\phi_{3}} \otimes \ket{\tau_{\eps}},\; 
    \ket{\phi_{4}} \otimes \ket{\tau_{\eps}}\right\} 
    \subset (\X_1\otimes\Y_1)\otimes(\X_2\otimes\Y_2),
  \end{equation}
  assuming a uniform distribution $p_1 = p_2 = p_3 = p_4 = 1/4$, is at most
  \begin{equation}
    \label{eq:probability-four-bell}
    \frac{1}{2}\left(1 + \sqrt{1 - \eps^2}\right).
  \end{equation}
\end{theorem}

\begin{proof}
  One may formulate a cone program corresponding to state discrimination by 
  PPT measurements along similar lines to the cone program for separable
  measurements, simply by replacing the cone $\Sep(\X:\Y)$ by the cone
  $\PPT(\X:\Y)$ of positive semidefinite operators whose partial transpose is
  positive semidefinite.
  This cone program is a semidefinite program, as discussed in
  \cite{Cosentino13}, by virtue of the fact that partial transpose mapping is
  linear.
  
  Consider the following operator:
  \begin{equation}
    H_{\varepsilon} = \frac{1}{8}
    \Bigl[
      \I_{\X_1 \otimes \Y_1} \otimes \tau_{\varepsilon}
      + \sqrt{1 - \varepsilon^2}\,
      \I_{\X_1 \otimes \Y_1} \otimes \pt_{\negsmallspace\X_2}(\phi_4)
      \Bigr] \in \Herm(\X_1\otimes\Y_1\otimes\X_2\otimes\Y_2).
  \end{equation}
  It holds that
  \begin{equation}
    \tr(H_{\eps}) = \frac{1}{2}\left( 1 + \sqrt{1-\eps^2} \right),
  \end{equation}
  so to complete the proof it suffices to prove that $H_{\varepsilon}$ is
  dual feasible for the (semidefinite) cone program representing the 
  PPT state discrimination problem under consideration.
  Dual feasibility will follow from the condition
  \begin{equation}
    (\pt_{\negsmallspace\X_1} \otimes \pt_{\negsmallspace\X_2})
    \Bigl(
    H_{\varepsilon} - \frac{1}{4}\,\phi_k\otimes\tau_{\varepsilon}\Bigr)
    \in\Pos(\X_1\otimes\Y_1\otimes\X_2\otimes\Y_2)
  \end{equation}
  (which is sufficient but not necessary for feasibility) for $k = 1,2,3,4$.
  One may observe that
  \begin{equation}
    \pt_{\negsmallspace\X_2}(\tau_{\varepsilon}) 
    + \frac{\sqrt{1-\varepsilon^2}}{2}\phi_4
    = \frac{1}{2}
    \begin{pmatrix}
      1 + \varepsilon & 0 & 0 & 0\\[1mm]
      0 & \frac{\sqrt{1 - \varepsilon^2}}{2} 
      & \frac{\sqrt{1 - \varepsilon^2}}{2} & 0\\[2mm]
      0 & \frac{\sqrt{1 - \varepsilon^2}}{2} 
      & \frac{\sqrt{1 - \varepsilon^2}}{2} & 0\\[2mm]
      0 & 0 & 0 & 1 - \varepsilon
    \end{pmatrix}
  \end{equation}
  is positive semidefinite, from which it follows that
  \begin{equation}
    (\pt_{\negsmallspace\X_1} \otimes \pt_{\negsmallspace\X_2})
    \Bigl(
    H_{\varepsilon} - \frac{1}{4}\,\phi_1\otimes\tau_{\varepsilon}\Bigr)
    = \frac{1}{4} \phi_4 \otimes \pt_{\negsmallspace\X_2}(\tau_{\varepsilon})
    + \frac{\sqrt{1 - \varepsilon^2}}{8} \I_{\X_1\otimes\Y_1} \otimes \phi_4
  \end{equation}
  is also positive semidefinite.
  A similar calculation holds for $k=2,3,4$, which completes the proof.
\end{proof}

\begin{remark}
  Similar to Theorem \ref{thm:three-bell}, one has that the upper bound
  obtained by Theorem \ref{thm:four-bell} is optimal for LOCC measurements,
  as it is the probability obtained using teleportation.
\end{remark}

\section{State discrimination and unextendable product sets}
\label{sec:upb}

In this section, we study the state discrimination problem for collections of
states formed by unextendable product sets.
An orthonormal collection of product vectors
\begin{equation}
  \A = \{ u_{k}\otimes v_{k} : k = 1, \ldots, N \} \subset \X \otimes \Y,
\end{equation}
for complex Euclidean spaces $\X=\complex^n$ and $\Y=\complex^m$, is said to be
an \emph{unextendable product set} if it is impossible to find a nonzero
product vector $u \otimes v \in \X \otimes \Y$ that is orthogonal to every
element of $\A$ \cite{BennettDMSST99a}.
That is, $\A$ is an unextendable product set if, for every choice of vectors
$u\in\X$ and $v\in\Y$ satisfying either $\ip{u}{u_{k}} = 0$ or 
$\ip{v}{v_{k}} = 0$ for each $k\in\{1,\ldots,N\}$, one has that either $u = 0$
or $v = 0$ (or both).

Two subsections follow.
The first subsection establishes a simple criterion for the states formed by
any unextendable product set to be perfectly discriminated by separable
measurements, and the second subsection proves that any set of states formed
by taking the union of an unextendable product set $\A \subset \X \otimes \Y$ 
together with any pure state $z \in \X\otimes\Y$ orthogonal to every element of
$\A$ cannot be perfectly discriminated by a separable measurement.
(It is evident that PPT measurements allow a perfect discrimination in both
cases.)

\subsubsection*{A criterion for perfect separable discrimination of an
unextendable product set}

Here we provide a simple criterion for when an unextendable product set can be
perfectly discriminated by separable measurements, and we use this criterion to
show that there is an unextendable product set $\A \subset \X\otimes\Y$ that is
not perfectly discriminated by any separable measurement when
$\X = \Y = \complex^4$.
It is known that no unextendable product set $\A \subset \X\otimes\Y$
spanning a proper subspace of $\X\otimes\Y$ can be perfectly discriminated by
an LOCC measurement \cite{BennettDMSST99a}, while every unextendable product
set can be discriminated perfectly by a PPT measurement.
It is also known that every unextendable product set $\A\subset\X\otimes\Y$ can
be perfectly discriminated by separable measurements in the case 
$\X = \Y = \complex^3$ \cite{DMSST03}.

The following notation will be used throughout this subsection.
For $\X=\complex^n$, $\Y=\complex^m$, and
$\A = \{ u_{k}\otimes v_{k} : k = 1, \ldots, N \} \subset \X\otimes\Y$
being an unextendable product set, we will write
$\A_k = \A \backslash \{u_k \otimes v_k\}$, and define a set of rank-one
product projections
\begin{equation}
\label{eq:P_k-sets}
\P_k = \bigl\{ x x^{\ast} \otimes y y^{\ast}\,:\,
x\in\X,\,y\in\Y,\,\norm{x} = \norm{y} =1,\;\text{and}\;
x\otimes y\perp \A_k\bigr\}
\end{equation}
for each $k = 1,\ldots,N$.
One may interpret each element $x x^{\ast} \otimes y y^{\ast}$ of 
$\P_k$ as corresponding to a product vector $x \otimes y$ that could replace
$u_k \otimes v_k$ in $\A$, yielding a (not necessarily unextendable)
orthonormal product set.

The following theorem states that the sets $\P_1,\ldots,\P_N$ defined above
determine whether or not an unextendable product set can be perfectly
discriminated by separable measurements.

\begin{theorem}\label{thm:upb_sep_characterize}
  Let $\X=\complex^n$ and $\Y=\complex^m$ be complex Euclidean spaces and let
  \begin{equation}
    \A = \{ u_{k}\otimes v_{k} : k = 1, \ldots, N \} \subset \X\otimes\Y
  \end{equation}
  be an unextendable product set.
  The following two statements are equivalent:
  \begin{mylist}{6mm}
  \item[1.]
    There exists a separable measurement $\{P_1,\ldots,P_N\}\subset\Sep(\X:\Y)$
    that perfectly discriminates the
    states represented by $\A$ (for any choice of nonzero probabilities
    $p_1,\ldots,p_N$).
  \item[2.]
    For $\P_1,\ldots,\P_N$ as defined in \eqref{eq:P_k-sets}, one has that the
    identity operator $\I_{\X}\otimes\I_{\Y}$ can be written as a nonnegative
    linear combination of projections in the set $\P_1\cup\cdots\cup \P_N$.
  \end{mylist}
\end{theorem}

\begin{proof}
  Assume first that statement 2 holds, so that one may write
  \begin{equation}
    \I_{\X}\otimes\I_{\Y} =
    \sum_{k = 1}^N \sum_{j = 1}^{M_k} \lambda_{k,j}\,
    x_{k,j} x_{k,j}^{\ast}\otimes y_{k,j} y_{k,j}^{\ast}
  \end{equation}
  for some choice of positive integers $M_1,\ldots,M_N$, nonnegative real
  numbers $\{\lambda_{k,j}\}$, and product vectors
  $\{x_{k,j} \otimes y_{k,j}\}$ satisfying
  \begin{equation}
    x_{k,j} x_{k,j}^{\ast}\otimes y_{k,j} y_{k,j}^{\ast} \in \P_k
  \end{equation}
  for each $k\in\{1,\ldots,N\}$ and $j \in \{1,\ldots,M_k\}$.
  Define
  \begin{equation} \label{eq:P_k-enumeration}
    P_k = \sum_{j = 1}^{M_k} \lambda_{k,j}\,
    x_{k,j} x_{k,j}^{\ast}\otimes y_{k,j} y_{k,j}^{\ast}
  \end{equation}
  for each $k\in\{1,\ldots,N\}$.
  It is clear that $\{P_1,\ldots,P_N\}$ is a separable measurement, and by the
  definition of the sets $\P_1,\ldots,\P_N$ it necessarily holds that
  $\ip{P_k}{u_\ell u_\ell^* \otimes v_\ell v_\ell^*} = 0$ when $k\not=\ell$.
  This implies that $\{P_1,\ldots,P_N\}$ perfectly discriminates the states
  represented by $\A$, and therefore implies that statement 1 holds.
  
  Now assume that statement 1 holds: there exists a separable measurement
  $\{P_1,\ldots,P_N\}$ that perfectly discriminates the states represented by
  $\A$.
  As each measurement operator $P_k$ is separable, it is possible to write
  \begin{equation}
    P_k = \sum_{j = 1}^{M_k} \lambda_{k,j}\,
    x_{k,j} x_{k,j}^{\ast}\otimes y_{k,j} y_{k,j}^{\ast}
  \end{equation}
  for some choice of nonnegative integers $\{M_k\}$, positive real numbers
  $\{\lambda_{k,j}\}$, and unit vectors
  $\{x_{k,j}\,:\,j=1,\ldots,M_k\}\subset\X$ and 
  $\{y_{k,j}\,:\,j=1,\ldots,M_k\}\subset\Y$.
  The assumption that this measurement perfectly discriminates $\A$ implies that
  $x_{k,j}\otimes y_{k,j} \perp \A_k$, and therefore
  $x_{k,j} x_{k,j}^{\ast} \otimes y_{k,j} y_{k,j}^{\ast} \in \P_k$, for each
  $k = 1,\ldots,N$ and $j = 1,\ldots,M_k$.
  As $P_1+\cdots+P_N = \I_{\X} \otimes \I_{\Y}$, it follows that statement 2
  holds.
\end{proof}

It is not immediately clear that Theorem~\ref{thm:upb_sep_characterize} is
useful for determining whether or not any particular unextendable product set
can be discriminated by separable measurements, but indeed it is.
What makes this so is the fact that each set $\P_k$ is necessarily finite, as
the following proposition establishes.

\begin{prop}\label{lem:upb_finite}
  Let $\X$ and $\Y$ be complex Euclidean spaces, let
  $\A = \{ u_{k}\otimes v_{k} : k = 1, \ldots, N \} \subset \X\otimes\Y$
  be an unextendable product set, and let $\P_1,\ldots,\P_N$ be as defined in
  \eqref{eq:P_k-sets}.
  The sets $\P_1,\ldots,\P_N$ are finite.
\end{prop}

\begin{proof}
  Assume toward contradiction that $\P_k$ is infinite for some choice of
  $k\in\{1,\ldots,N\}$.
  There are finitely many subsets $S\subseteq \{1,\ldots,k-1,k+1,\ldots,N\}$,
  so  there must exist at least one such subset $S$ with the property that
  there are infinitely many pairwise nonparallel product vectors of the form
  $x\otimes y$ such that $x \perp u_j$ for every $j\in S$ and $y\perp v_j$ for
  every $j\not\in\S$.
  This implies that both the subspace of $\X$ orthogonal to
  $\{u_j\,:\,j\in S\}$ and the subspace of $\Y$ orthogonal to
  $\{v_j\,:\,j\not\in S\}$ have dimension at least 1, and
  at least one of them has dimension at least~2.
  It follows that there must exist a unit product vector $x \otimes y$ with
  three properties:
  (i) $x \perp u_j$ for every $j\in S$,
  (ii) $y \perp v_j$ for every $j\not\in\S$, and
  (iii) $x\otimes y\perp u_k \otimes v_k$.
  This contradicts the fact that $\A$ is unextendable, and therefore completes
  the proof.
\end{proof}

Given Proposition~\ref{lem:upb_finite}, it becomes straightforward to make use
of Theorem~\ref{thm:upb_sep_characterize} computationally.
The sets $\P_1,\ldots,\P_N$ can be computed by iterating over all 
$S \subseteq \{1,\ldots,k-1,k+1,\ldots,N\}$ and
finding the (at most one) product state orthogonal to $\{ u_j : j \in S \}$ on
$\X$ and $\{ v_j : j \notin S \}$ on $\Y$. 
Then, the second statement in Theorem~\ref{thm:upb_sep_characterize} can be
checked through the use of linear programming (and even by hand in some cases).

\subsubsection*{Example}
We now present an example of an unextendable product set in $\X\otimes\Y$,
for $\X = \Y = \complex^4$, that cannot be perfectly discriminated by separable
measurements. 
In particular, let $\A$ be the unextendable product set consisting of $8$
states that were found in \cite{Fen06}:
\begin{equation}
  \begin{array}{ll}
    \ket{\phi_1} = \ket{0}\ket{0}, &
    \ket{\phi_5} = \left(\ket{1} + \ket{2} + \ket{3}\right)\left(\ket{0}
    - \ket{1} + \ket{2}\right)/3, \\
    \ket{\phi_2} = \ket{1}\left(\ket{0} - \ket{2} + \ket{3}\right)/\sqrt{3}, 
    \quad &
    \ket{\phi_6} = \left(\ket{0} - \ket{2} + \ket{3}\right)\ket{2}/\sqrt{3}, \\
    \ket{\phi_3} = \ket{2}\left(\ket{0} + \ket{1} - \ket{3}\right)/\sqrt{3}, &
    \ket{\phi_7} = \left(\ket{0} + \ket{1} - \ket{3}\right)\ket{1}/\sqrt{3}, \\
    \ket{\phi_4} = \ket{3}\ket{3}, &
    \ket{\phi_8} = \left(\ket{0} - \ket{1} + \ket{2}\right)\left(\ket{1}
    + \ket{2} + \ket{3}\right)/3.
  \end{array}
\end{equation}

For each $k = 1, \ldots, 8$, there are exactly $6$ product states contained in
$\A_k$ for each choice of $k$, which we represent by product vectors
$\ket{\phi_{k,j}}$ for $j = 1, \ldots, 6$.
To be explicit, these states are as follows (where we have omitted
normalization factors for brevity):\vspace{3mm}

\noindent\hspace{8pt}
\begin{minipage}{0.46\textwidth}
  $\ket{\phi_{1,1}} = \ket{0}\ket{0}$,\\
  $\ket{\phi_{1,2}} = \left(\ket{0} + \ket{1} - \ket{3}\right)
  \left(\ket{0} + \ket{2}\right)$,\\
  $\ket{\phi_{1,3}} = \left(\ket{0} - \ket{1}\right)
  \left(\ket{0} + \ket{1} - \ket{3}\right)$,
\end{minipage}\hfill
\begin{minipage}{0.46\textwidth}
  $\ket{\phi_{1,4}} = \left(\ket{0} - \ket{1} + 
  \ket{2}\right)\left(\ket{0} - \ket{1} + \ket{2}\right)$, \\
  $\ket{\phi_{1,5}} = \left(\ket{0} + \ket{2}\right)
  \left(\ket{0} - \ket{2} + \ket{3}\right)$,\\
  $\ket{\phi_{1,6}} = \left(\ket{0} - \ket{2} + \ket{3}\right)
  \left(\ket{0} - \ket{1}\right)$,
\end{minipage}

\vspace{2mm}

\noindent\hspace{8pt}
\begin{minipage}{0.46\textwidth}
  $\ket{\phi_{2,1}} = \ket{1}\left(\ket{0} - \ket{2} + \ket{3}\right)$,\\
  $\ket{\phi_{2,2}} = \left(\ket{0} + \ket{1} - \ket{3}\right)\ket{2}$,\\
  $\ket{\phi_{2,3}} =  \left(\ket{0} + \ket{1}\right)\ket{3}$,
\end{minipage}\hfill
\begin{minipage}{0.46\textwidth}
  $\ket{\phi_{2,4}} = \left(\ket{0} - \ket{1} + \ket{2}\right)
  \left(\ket{1} - 2\ket{2} + \ket{3}\right)$, \\
  $\ket{\phi_{2,5}} = \left(\ket{1} + \ket{2} + \ket{3}\right)
  \left(\ket{0} - \ket{1} - 2\ket{2}\right)$, \\
  $\ket{\phi_{2,6}} = \left(\ket{1} - \ket{3}\right)\ket{0}$,
\end{minipage}

\vspace{2mm}

\noindent\hspace{8pt}
\begin{minipage}{0.46\textwidth}   
  $\ket{\phi_{3,1}} = \ket{2}\left(\ket{0} + \ket{1} - \ket{3}\right)$,\\
  $\ket{\phi_{3,2}} = \left(\ket{0} - \ket{2} + \ket{3}\right)\ket{1}$,\\
  $\ket{\phi_{3,3}} =  \left(\ket{2} - \ket{3}\right)\ket{0}$,
\end{minipage}\hfill
\begin{minipage}{0.46\textwidth}
  $\ket{\phi_{3,4}} = \left(\ket{1} + \ket{2} + \ket{3}\right)
  \left(\ket{0} + 2\ket{1} + \ket{2}\right)$, \\
  $\ket{\phi_{3,5}} = \left(\ket{0} - \ket{1} + \ket{2}\right)
  \left(2\ket{1} - \ket{2} - \ket{3}\right)$, \\
  $\ket{\phi_{3,6}} = \left(\ket{0} - \ket{2}\right)\ket{3}$,
\end{minipage}

\vspace{2mm}

\noindent\hspace{8pt}
\begin{minipage}{0.46\textwidth}
  $\ket{\phi_{4,1}} = \ket{3}\ket{3}$,\\
  $\ket{\phi_{4,2}} = \left(\ket{0} + \ket{1} - \ket{2}\right)
  \left(\ket{2} + \ket{3}\right)$,\\
  $\ket{\phi_{4,3}} =  \left(\ket{1} + \ket{3}\right)
  \left(\ket{0} + \ket{1} - \ket{3}\right)$,
\end{minipage}\hfill
\begin{minipage}{0.46\textwidth}
  $\ket{\phi_{4,4}} = \left(\ket{2} + \ket{3}\right)
  \left(\ket{0} - \ket{2} + \ket{3}\right)$, \\
  $\ket{\phi_{4,5}} = \left(\ket{1} + \ket{2} + \ket{3}\right)
  \left(\ket{1} + \ket{2} + \ket{3}\right)$, \\
  $\ket{\phi_{4,6}} = \left(\ket{0} - \ket{2} + \ket{3}\right)
  \left(\ket{1} + \ket{3}\right)$,
\end{minipage}

\vspace{2mm}

\noindent\hspace{8pt}
\begin{minipage}{0.46\textwidth}
  $\ket{\phi_{5,1}} = \left(\ket{1} + \ket{2} + \ket{3}\right)
  \left(\ket{0} - \ket{1} + \ket{2}\right)$,\\
  $\ket{\phi_{5,2}} = \ket{1}\left(2\ket{0} + \ket{2} - \ket{3}\right)$,\\
  $\ket{\phi_{5,3}} = \ket{3}\ket{0}$,
\end{minipage}\hfill
\begin{minipage}{0.46\textwidth}
  $\ket{\phi_{5,4}} = \left(\ket{0} - \ket{2} - 2\ket{3}\right)\ket{2}$, \\
  $\ket{\phi_{5,5}} = \ket{2}\left(2\ket{0} - \ket{1} + \ket{3}\right)$, \\
  $\ket{\phi_{5,6}} = \left(\ket{0} + \ket{1} + 2\ket{3}\right)\ket{1}$,
\end{minipage}

\vspace{2mm}

\noindent\hspace{8pt}
\begin{minipage}{0.46\textwidth}
  $\ket{\phi_{6,1}} = \left(\ket{0} - \ket{2} + \ket{3}\right)\ket{2}$,\\
  $\ket{\phi_{6,2}} = \ket{3}\left(\ket{0} - \ket{2}\right)$,\\
  $\ket{\phi_{6,3}} =  \ket{0}\left(\ket{2} - \ket{3}\right)$,
\end{minipage}\hfill
\begin{minipage}{0.46\textwidth}
  $\ket{\phi_{6,4}} = \left(\ket{0} - \ket{1} - 2\ket{2}\right)
  \left(\ket{1} + \ket{2} + \ket{3}\right)$, \\
  $\ket{\phi_{6,5}} = \ket{2}\left(\ket{0} - \ket{2} + \ket{3}\right)$, \\
  $\ket{\phi_{6,6}} = \left(\ket{1} - 2\ket{2} + \ket{3}\right)
  \left(\ket{0} - \ket{1} + \ket{2}\right)$,
\end{minipage}

\vspace{2mm}

\noindent\hspace{8pt}
\begin{minipage}{0.46\textwidth}
  $\ket{\phi_{7,1}} = \left(\ket{0} + \ket{1} - \ket{3}\right)\ket{1}$,\\
  $\ket{\phi_{7,2}} = \ket{0}\left(\ket{1} - \ket{3}\right)$,\\
  $\ket{\phi_{7,3}} =  \ket{1}\left(\ket{0} + \ket{1} - \ket{3}\right)$,
\end{minipage}\hfill
\begin{minipage}{0.46\textwidth}
  $\ket{\phi_{7,4}} = \left(\ket{0} + 2\ket{1} + \ket{2}\right)
  \left(\ket{1} + \ket{2} + \ket{3}\right)$, \\
  $\ket{\phi_{7,5}} = \left(2\ket{1} - \ket{2} - \ket{3}\right)
  \left(\ket{0} - \ket{1} + \ket{2}\right)$, \\
  $\ket{\phi_{7,6}} = \ket{3}\left(\ket{0} + \ket{1}\right)$,
\end{minipage}

\vspace{2mm}

\noindent\hspace{8pt}
\begin{minipage}{0.46\textwidth}
  $\ket{\phi_{8,1}} = \left(\ket{0} - \ket{1} + \ket{2}\right)
  \left(\ket{1} + \ket{2} + \ket{3}\right)$,\\
  $\ket{\phi_{8,2}} = \ket{1}\left(\ket{0} - \ket{2} - 2\ket{3}\right)$,\\
  $\ket{\phi_{8,3}} =  \left(2\ket{0} - \ket{1} + \ket{3}\right)\ket{1}$,
\end{minipage}\hfill
\begin{minipage}{0.46\textwidth}
  $\ket{\phi_{8,4}} = \ket{0}\ket{3}$, \\
  $\ket{\phi_{8,5}} = \left(2\ket{0} + \ket{2} - \ket{3}\right)\ket{2}$, \\
  $\ket{\phi_{8,6}} = \ket{2}\left(\ket{0} + \ket{1} + 2\ket{3}\right)$.
\end{minipage}

\vspace{3mm}

\noindent 
One may verify by a computer that $\I\otimes\I$ is not contained in the convex
cone generated by 
\begin{equation}
  \label{eq:Feng-replacement-set}
  \bigl\{ \ket{\phi_{k,j}}\bra{\phi_{k,j}}\,:\,k = 1,\ldots,8,\;j=1,\ldots,6
  \bigr\}.
\end{equation}
(In fact, $\I\otimes\I$ is not in the linear span of the set
\eqref{eq:Feng-replacement-set}.)
Theorem~\ref{thm:upb_sep_characterize} therefore implies that this unextendable
product set is not perfectly discriminated by separable measurements.

\subsubsection*{Impossibility to discriminate an unextendable product set
plus one more pure state}

Next, we prove an upper bound on the probability to correctly discriminate
any unextendable product set, together with one extra pure state orthogonal
to the members of the unextendable product set, by a separable measurement.
Central to the proof of this statement is a family of positive linear maps
previously studied in the literature \cite{Terhal2001, Bandyopadhyay2005}.

Before proving this fact, we note that it is fairly straightforward to obtain
a qualitative result along similar lines:
if a separable measurement were able to perfectly discriminate a particular
product set from any state orthogonal to this product set, there would
necessarily be a separable measurement operator orthogonal to the space spanned
by the product set, implying that some nonzero product state must be orthogonal
to the product set (and therefore the product set must be extendable).
Related results based on this sort of argument may be found in
\cite{Bandyopadhyay2011}.
An advantage of the method described in the present paper is that one obtains
precise bounds on the optimal discrimination probability, as opposed to a
statement that a perfect discrimination is not possible.

The following lemma is required for the proof of the theorem below.

\begin{lemma}[Terhal]
  \label{lemma:lambda}
  For given complex Euclidean spaces $\X=\complex^n$ and $\Y=\complex^m$, and
  any unextendable product set
  \begin{equation}
    \A = \{ u_{k}\otimes v_{k} : k = 1, \ldots, N \} \subset \X\otimes\Y,
  \end{equation}
  there exists a positive real number $\lambda_{\A} > 0$ such that
  \begin{equation}
    \left(\I_{\X} \otimes y^{\ast}\right)
    \left( \sum_{k = 1}^{N}u_{k}u_{k}^{\ast}\otimes v_{k}v_{k}^{\ast} \right)
    \left(\I_{\X} \otimes y\right)
    - \lambda_{\A}\norm{y}^{2}\I_{\X} \in \Pos(\X),
  \end{equation}
  for every $y \in \Y$.
\end{lemma}

\noindent
A proof of the lemma, as well as a constructive procedure
to calculate a bound on $\lambda_{\A}$, can be found in \cite{Terhal2001}.

\begin{theorem}
  \label{thm:upb}
  Let $\X = \complex^n$ and $\Y=\complex^m$ be complex Euclidean spaces, let
  \begin{equation}
    \A = \{ u_{k}\otimes v_{k} : k = 1, \ldots, N \} \subset \X\otimes\Y
  \end{equation}
  be an unextendable product set, and let $z \in \X\otimes\Y$ be a unit vector
  orthogonal to $\A$.
  Assuming a uniform selection, the probability to correctly discriminate the
  states corresponding to the set $\A \cup \{ z \}$ by a separable measurement,
  assuming a uniform selection of states, is upper-bounded by 
  \begin{equation}
    1 - \frac{\lambda_{\A}}{(N + 1)\delta},
  \end{equation}
  where $\lambda_{\A}$ is a positive real number satisfying the requirements 
  of Lemma~\ref{lemma:lambda} and $\delta = \norm{\tr_{\X}(z z^{\ast})}$.
\end{theorem}

\begin{proof}
  Consider the following Hermitian operator:
  \begin{equation}
    H = \frac{1}{N+1} \left( \sum_{k = 1}^{N}u_{k}u_{k}^{\ast}\otimes 
    v_{k}v_{k}^{\ast} + \left( 1- \frac{\lambda_{\A}}{\delta} \right) 
    zz^{\ast} \right).
  \end{equation}
  We want to show that $H$ is a feasible solution of the dual problem
  \eqref{eq:dual-problem} for the state discrimination problem under
  consideration.
  It is clear that
  \begin{equation}
    H - \frac{1}{N+1}u_{k}u_{k}^{\ast}\otimes v_{k}v_{k}^{\ast} \in 
    \Pos(\X\otimes\Y) \subset \BPos(\X:\Y)
  \end{equation}
  for $k = 1, \ldots, N$. 
  The remaining constraint left to be checked is the following:
  \begin{equation}
    H - \frac{1}{N+1}zz^{\ast} = 
    \frac{1}{N+1}\left(
    \sum_{k = 1}^{N}u_{k}u_{k}^{\ast}\otimes v_{k}v_{k}^{\ast} -
    \frac{\lambda_{\A}}{\delta} zz^{\ast} \right) \in \BPos(\X:\Y).
  \end{equation}
  Using the fact that
  \begin{equation}
    \label{eq:upb_constraint} 
    \delta\norm{y}^{2}\I_{\X} - 
    \left(\I_{\X} \otimes y^{\ast}\right)zz^{\ast}\left(\I_{\X} \otimes y\right)
    \in \Pos(\X),
  \end{equation} 
  for any $y \in \Y$, together with Lemma \ref{lemma:lambda}, one has that
  \begin{equation}
    \left( \I \otimes y \right)^{\ast}
    \left(\sum_{k = 1}^{N}u_{k}u_{k}^{\ast}\otimes v_{k}v_{k}^{\ast} -
    \frac{\lambda_{\A}}{\delta} zz^{\ast}\right)
    \left( \I \otimes y \right) \in \Pos(\X)
  \end{equation}
  and therefore the constraint \eqref{eq:upb_constraint} is satisfied.
  Finally, it holds that
  \begin{equation}
    \tr(H) = 1 - \frac{\lambda_{\A}}{(N + 1)\delta},
  \end{equation}
  which completes the proof.
\end{proof}

\subsubsection*{Example}

Theorem \ref{thm:upb} allow us to find specific bounds for the probability 
of correctly discriminating certain sets of states.
For instance, here we consider the following unextendable product set in 
$\X\otimes\Y$ for $\X = \Y = \complex^3$, commonly known as the 
\emph{tiles set}:
\begin{equation}
  \begin{array}{llll}
    \ket{\phi_1} = \ket{0}\left(\frac{\ket{0}-\ket{1}}{\sqrt{2}}\right),
    &\ket{\phi_2} = \ket{2}\left(\frac{\ket{1}-\ket{2}}{\sqrt{2}}\right),
    &\ket{\phi_3} = \left(\frac{\ket{0}-\ket{1}}{\sqrt{2}}\right)\ket{2},
    &\ket{\phi_4} = \left(\frac{\ket{1}-\ket{2}}{\sqrt{2}}\right)\ket{0},\\[4mm]
    \multicolumn{4}{c}{\ket{\phi_5} = 
      \frac{1}{3}\left(\ket{0}+\ket{1}+\ket{2}\right)
      \left(\ket{0}+\ket{1}+\ket{2}\right).}
  \end{array}
\end{equation}
For a pure state orthogonal to this set, one may take
\begin{equation}
  \ket{\psi} = \frac{1}{2}\left(\ket{0}\ket{0} + \ket{0}\ket{1} -
  \ket{0}\ket{2} -  \ket{1}\ket{2}\right).
\end{equation}
Using the procedure described in \cite{Terhal2001}, one obtains
\begin{equation}
  \lambda_{\A} \geq \frac{1}{9}\left( 1 - \sqrt{\frac{5}{6}} \right)^{2}.
\end{equation}
Therefore, if we assume that each state is selected with probability $1/6$, 
the maximum probability of correctly discriminating the set 
$\left\{ \ket{\phi_1}, \ldots, \ket{\phi_5}, \ket{\psi} \right\}$ by a
separable measurement is at most
\begin{equation}
  1 - \frac{1}{54}\frac{\left( 1 - \sqrt{\frac{5}{6}} \right)^{2}}
  {\cos\left(\frac{\pi}{8}\right)^{2}} <  1 - 1.647 \times 10^{-4}.
\end{equation}

\section{An optimal bound on discriminating the Yu--Duan--Ying states}
\label{sec:ydy}

In this section we prove a tight bound of $3/4$ on the maximum success
probability for any LOCC measurement to discriminate the set of states 
\eqref{eq:ydy_states} exhibited by Yu, Duan, and Ying \cite{YuDY12},
assuming a uniform selection of states.
The fact that this bound can be achieved by an LOCC measurement is trivial:
if Alice and Bob measure their parts of the states with respect to the standard
basis, they can easily discriminate $\ket{\phi_1}$, $\ket{\phi_2}$, and 
$\ket{\phi_4}$, erring only in the case that they receive $\ket{\phi_3}$.
The fact that this bound is optimal will be proved by exhibiting a 
feasible solution $H$ to the dual problem \eqref{eq:dual-problem},
corresponding to the state discrimination problem at hand, such that
$\tr(H) = 3/4$.

It is convenient for the analysis that follows to make use of the
correspondence between operators and vectors given by the linear function
defined by the action
\begin{equation}
  \op{vec}(\ket{k}\bra{j}) = \ket{k} \ket{j}
\end{equation}
on standard basis vectors.
With respect to this correspondence, the states \eqref{eq:ydy_states} are given
by tensor products of the Pauli operators \eqref{eq:Pauli-operators} as follows:
\begin{equation}
  \ket{\phi_1} = \frac{1}{2}\op{vec}(U_1),\quad
  \ket{\phi_2} = \frac{1}{2}\op{vec}(U_2),\quad
  \ket{\phi_3} = \frac{1}{2}\op{vec}(U_3),\quad\text{and}\quad
  \ket{\phi_4} = \frac{1}{2}\op{vec}(U_4),\quad
\end{equation}
for
\begin{alignat}{2}
  U_1 & = 
  \sigma_0\otimes\sigma_0 = 
  \begin{pmatrix}
    1 & 0 & 0 & 0 \\
    0 & 1 & 0 & 0 \\
    0 & 0 & 1 & 0 \\
    0 & 0 & 0 & 1
  \end{pmatrix}, & 
  U_2 & = 
  \sigma_1\otimes\sigma_1 =
  \begin{pmatrix}
    0 & 0 & 0 & 1 \\
    0 & 0 & 1 & 0 \\
    0 & 1 & 0 & 0 \\
    1 & 0 & 0 & 0
  \end{pmatrix},\\[3mm]
  U_3 & = 
  i \sigma_2 \otimes \sigma_1 =
  \begin{pmatrix}
    0 & 0 & 0 & 1 \\
    0 & 0 & 1 & 0 \\
    0 & -1 & 0 & 0 \\
    -1 & 0 & 0 & 0
  \end{pmatrix}, \qquad & 
  U_4 & = 
  \sigma_3 \otimes \sigma_1 =
  \begin{pmatrix}
    0 & 1 & 0 & 0 \\
    1 & 0 & 0 & 0 \\
    0 & 0 & 0 & -1 \\
    0 & 0 & -1 & 0
  \end{pmatrix}.
\end{alignat}
A feasible solution of the dual problem \eqref{eq:dual-problem} is based on
a construction of block positive operators that correspond, via the Choi
isomorphism, to the family of positive maps introduced by Breuer and Hall
\cite{Breuer06,Hall06}.

\begin{prop}[Breuer--Hall]
  \label{prop:breuer-hall}
  Let $\X = \Y = \complex^n$ and let $U,V\in\Unitary(\Y,\X)$ be unitary
  operators such that $U^{\t}V \in \Unitary(\Y)$ is skew-symmetric:
  $(V^{\t}U)^{\t} = -V^{\t}U$.
  It holds that
  \begin{equation}
    \I_{\X}\otimes\I_{\Y} - \vec(U) \vec(U)^{\ast} - 
    	\pt_{\negsmallspace\X}(\vec(V)\vec(V)^{\ast})
    \in \BPos(\X:\Y).
  \end{equation}
\end{prop}
\begin{proof}
  For every unit vector $y\in\Y$, one has
  \begin{multline} \label{eq:BH}
    \qquad
    (\I_{\X} \otimes y^{\ast})
    (\I_{\X}\otimes\I_{\Y} - \vec(U) 
    \vec(U)^{\ast} - \pt_{\negsmallspace\X}(\vec(V)\vec(V)^{\ast}))
    (\I_{\X} \otimes y) \\
    = \I_{\X} - U \overline{y} y^{\t} U^{\ast} - \overline{V}y
    y^{\ast}V^{\t}.
    \qquad
  \end{multline}
  As it holds that $V^{\t}U$ is skew-symmetric, we have
  \begin{equation}
    \bigip{\overline{V}y}{U\overline{y}}
    = y^{\ast} V^{\t}U \overline{y} 
    = \bigip{y y^{\t}}{V^{\t}U}
    = 0,
  \end{equation}
  as the last inner product is between a symmetric and a skew-symmetric
  operator.
  Because $U$ and $V$ are unitary, it follows that
  $U\overline{y}y^{\t}U^{\ast} + \overline{V}yy^{\ast}V^{\t}$ is a rank two
  orthogonal projection, so the operator represented by \eqref{eq:BH} is also a
  projection and is therefore positive semidefinite.
\end{proof}

\begin{remark}
  The assumption of Proposition \ref{prop:breuer-hall} requires $n$ to be even,
  as skew-symmetric unitary operators exist only in even dimensions.
\end{remark}

Now, for one of the four states $\rho_1 = \ket{\phi_1}\bra{\phi_1}$,
$\rho_2 = \ket{\phi_2}\bra{\phi_2}$,
$\rho_3 = \ket{\phi_3}\bra{\phi_3}$, or
$\rho_4 = \ket{\phi_4}\bra{\phi_4}$
drawn with uniform probabilities $p_{1} = \cdots = p_{4} = 1/4$, one has that
the following operator is a feasible solution to the dual problem
\eqref{eq:dual-problem}:
\begin{equation}
  \label{eq:H_2}
  H = \frac{1}{16}(\I_{\X}\otimes\I_{\Y} - 
  \pt_{\negsmallspace\X}(\vec(V)\vec(V)^{\ast}))
\end{equation}
for
\begin{equation}
  V = i \sigma_2 \otimes \sigma_3
  = \begin{pmatrix}
    0 & 0 & 1 & 0\\
    0 & 0 & 0 & -1\\
    -1 & 0 & 0 & 0\\
    0 & 1 & 0 & 0
  \end{pmatrix}.
\end{equation}
Due to Proposition~\ref{prop:breuer-hall}, the feasibility of $H$ follows from
the condition
\begin{equation}
  (V^{\t} U_k)^{\t} = - V^{\t} U_k,
\end{equation}
which can be checked by inspecting each of the four cases.
It is easy to calculate that $\tr(H) = 3/4$, and so the required bound has been
obtained.

\section{Conclusion}

In this paper we have used techniques from convex optimization, and cone
programming in particular, to study the limitations of separable measurements
for the task of discriminating sets of bipartite state.

Several interesting questions regarding the discrimination of sets of 
bipartite states by means of separable and LOCC measurements remain unsolved.
Among them are the following two questions.

\begin{itemize}
\item 
  In Section \ref{sec:bell} we proved a tight bound on the entanglement
  cost of discriminating sets of Bell states by means of LOCC protocols. 
  More generally, one could ask how much entanglement it costs to distinguish 
  maximally entangled states in $\complex^{n}\otimes\complex^{n}$.

\item Ghosh et al.~\cite{GhoshKRS04} have shown that orthogonal maximally
  entangled states, which are in canonical form, can always be discriminated,
  by means of LOCC protocols, if two copies of each of the states are provided.
  The question of whether two copies are sufficient to discriminate any set 
  of orthogonal pure states is open even for separable and PPT measurements.
\end{itemize}

The techniques presented in the paper are not intrinsically limited 
to the setting of bipartite pure states---applications of these techniques to
mixed states and multipartite states are topics for possible future work.

\subsection*{Acknowledgments}

AC thanks Jamie Sikora for insightful discussions on cone programming.
SB thanks the Institute for Quantum Computing at the University of Waterloo for
supporting a visit during which some of the research reported in this paper was
done.
We also thank Ajit Iqbal Singh and Sibasish Ghosh for catching a mistake
in the example following Proposition~8 in the journal version of this paper,
which has been corrected in this manuscipt.

SB is supported by DST-SERB project SR/S2/LOP-18/2012;
AC is supported by NSERC, the US Army Research Office, and a David R. Cheriton
Graduate Scholarship;
NJ is supported by NSERC;
VR is supported by NSERC, the US Army Research Office, and a David R. Cheriton
Graduate Scholarship;
JW is supported by NSERC; and
NY is supported by NSERC.

\bibliographystyle{alpha}

\newcommand{\etalchar}[1]{$^{#1}$}

\end{document}